\documentclass[letterpaper, 10 pt, conference]{ieeeconf}  %

\IEEEoverridecommandlockouts                              %
\usepackage{url}           %
\usepackage[hidelinks]{hyperref} %

\usepackage[
  backend=biber,
  style=numeric-comp,   %
  sortcites=true,
  sorting=none,         %
  natbib=true,          %
  giveninits=true,
  date=year,           %
  maxcitenames=4,       %
  minbibnames=3,        %
  maxbibnames=4,        %
  doi=false,            %
  url=false,
  hyperref=true
]{biblatex}
\DeclareSourcemap{
  \maps[datatype=bibtex]{
    \map{
      \step[fieldset=file,      null]
      \step[fieldset=abstract,  null]
      \step[fieldset=keywords,  null]
      \step[fieldset=urldate,   null]
      \step[fieldset=month,     null]
      \step[fieldset=day,       null]
      \step[fieldset=language,  null]
      \step[fieldset=isbn,      null]
      \step[fieldset=issn,      null]
      \step[fieldsource=doi,
            match=\regexp{^\s*https?://(dx\.)?doi\.org/},
            replace={}]
      \step[fieldsource=booktitle,
            match=\regexp{^\s*Proceedings\s+of\s+(the\s+)?},
            replace={}]
      \step[fieldsource=eventtitle,
            match=\regexp{^\s*Proceedings\s+of\s+(the\s+)?},
            replace={}]
      \step[fieldsource=booktitle,
            match=\regexp{^\s*(19|20)\d{2}\s+},
            replace={}]
      \step[fieldsource=eventtitle,
            match=\regexp{^\s*(19|20)\d{2}\s+},
            replace={}]
      \step[fieldsource=booktitle,
            match=\regexp{^\s*\d+(st|nd|rd|th)\s+},
            replace={}]
      \step[fieldsource=eventtitle,
            match=\regexp{^\s*\d+(st|nd|rd|th)\s+},
            replace={}]
    }
  }
}
\renewbibmacro{in:}{}

\DeclareFieldFormat{title}{\mkbibemph{#1}}
\DeclareFieldFormat[article,incollection,techreport,inproceedings,book,misc]{title}{\mkbibquote{#1}}

\DeclareFieldFormat{eprint:arxiv}{arXiv\addcolon\space#1}


\usepackage{amsthm}
\newtheorem{theorem}{Theorem}
\newtheorem{lemma}{Lemma}
\newtheorem{proposition}{Proposition}

\newtheorem{definition}{Definition}
\newtheorem{remark}{Remark}

\usepackage[dvipsnames]{xcolor}
\definecolor{porange}{HTML}{E77500} %
\usepackage{multicol}
\usepackage{amsfonts}
\usepackage{fancyhdr}
\usepackage{comment}
\usepackage{amsmath}
\usepackage{changepage}
\usepackage{amssymb}
\usepackage{graphicx}
\usepackage{adjustbox}
\usepackage{latexsym}
\usepackage{enumerate}
\usepackage{cleveref}
\usepackage{float}
\usepackage{subcaption}

\usepackage{mathtools}
\usepackage[ruled,vlined,linesnumbered]{algorithm2e}
\usepackage{bbm}
\usepackage{bm}
\usepackage{booktabs}
\usepackage[xindy, acronym, nowarn]{glossaries}   %
\usepackage{multirow}
\usepackage{makecell}
\usepackage{colortbl}
\usepackage{balance}

\usepackage{textcomp}
\usepackage{stfloats}  %

\usepackage{fontawesome5}  %

\usepackage{algorithmic}

\definecolor{porange}{HTML}{E77500} %

\DeclareDocumentEnvironment{example}{}{\noindent\textbf{Running example:}\itshape}{}

\usepackage{ifthen}
\newboolean{include-notes}
\newboolean{highlight-new}
\newboolean{include-remove}
\setboolean{include-notes}{true}
\setboolean{highlight-new}{true}
\setboolean{include-remove}{true}

\usepackage[dvipsnames]{xcolor}
\usepackage[normalem]{ulem}
\newcommand{\jaime}[1]{\ifthenelse{\boolean{include-notes}}{\textcolor{orange}{\textbf{Jaime:} #1}}{}}
\newcommand{\haimin}[1]{\ifthenelse{\boolean{include-notes}}{\textcolor{magenta}{\textbf{Haimin:} #1}}{}}
\newcommand{\justin}[1]{\ifthenelse{\boolean{include-notes}}{\textcolor{Cerulean}{\textbf{Justin:} #1}}{}}
\newcommand{\himani}[1]{\ifthenelse{\boolean{include-notes}}{\textcolor{Plum}{\textbf{Nishanth:} #1}}{}}
\newcommand{\david}[1]{\ifthenelse{\boolean{include-notes}}{\textcolor{teal}{\textbf{David:} #1}}{}}
\newcommand{\duy}[1]{\ifthenelse{\boolean{include-notes}}{\textcolor{blue}{\textbf{Duy:} #1}}{}}
\newcommand{\remove}[1]{\ifthenelse{\boolean{include-remove}}{\textcolor{red}{\sout{#1}}}{}}
\newcommand{\new}[1]{\ifthenelse{\boolean{highlight-new}}{\textcolor{blue}{#1}}{#1}}
\newcommand{\todo}[1]{\ifthenelse{\boolean{include-notes}}{\textcolor{blue}{\textbf{TODO:} #1}}{}}
\newcommand{\guy}[1]{\ifthenelse{\boolean{include-notes}}{\textcolor{Cerulean}{\textbf{Guy:} #1}}{}}
\newcommand{\flag}[1]{\ifthenelse{\boolean{include-notes}}{\textcolor{red}{#1}}{#1}}

\newcommand{\princeton}[1]{\ifthenelse{\boolean{include-notes}}{\textcolor{orange}{#1}}{}}

\newcommand{\p}[1]{\smallskip \noindent \textbf{{#1}.}}

\newcommand{\ie}{\textit{i.e.}}

\DeclareMathOperator*{\argmaxB}{argmax}
\DeclareMathOperator*{\argminB}{argmin}

\newcommand{\failureset}{{\mathcal{F}}}
\newcommand{\safeset}{{\Omega}}
\newcommand{\task}{{\text{task}}}

\newcommand{\shield}{\text{\tiny{\faShield*}}}

\newcommand{\fallback}{{\policy^\shield}}

\newcommand{\cbf}{h}

\newcommand{\reals}{\mathbb{R}}

\newcommand{\compl}{\mathsf{c}}

\DeclareMathOperator*{\expectation}{\mathbb{E}}

\DeclareMathOperator{\sgn}{sgn}

\newcommand{\st}{\textnormal{s.t.}}

\newcommand{\state}{{x}}

\newcommand{\ctrl}{{u}}

\newcommand{\dstb}{{d}}

\newcommand{\xset}{{\mathcal{X}}}  %

\newcommand{\maxsafeset}{\Omega^*}
\newcommand{\cset}{{\mathcal{U}}}
\newcommand{\dset}{{\mathcal{D}}}

\newcommand{\dyn}{{f}}

\newcommand{\margin}{{g}}
\newcommand{\valfunc}{{V}}
\newcommand{\qfunc}{{\mathcal{Q}}}
\newcommand{\policy}{{\pi}}
\newcommand{\cpolicy}{\pi^\ctrl}
\newcommand{\dpolicy}{\pi^\dstb}

\newcommand{\buffer}{\mathcal{B}}

\newglossarystyle{mylist}{%
  \setglossarystyle{list}%
}

\glsdisablehyper
\makeglossaries

\newglossaryentry{GDA}
{
  name={GDA},
  description={gradient descent-ascent},
  first={gradient descent-ascent (\glsentrytext{GDA})}
}

\newglossaryentry{RL}
{
  name={RL},
  description={reinforcement learning},
  first={reinforcement learning (\glsentrytext{RL})}
}

\newglossaryentry{HJ}
{
  name={HJ},
  description={Hamilton--Jacobi},
  first={Hamilton--Jacobi (\glsentrytext{HJ})}
}

\newglossaryentry{HJI}
{
  name={HJI},
  description={Hamilton--Jacobi--Isaacs},
  first={Hamilton--Jacobi--Isaacs (\glsentrytext{HJI})}
}

\newglossaryentry{DCBF}
{
  name={DCBF},
  plural={DCBFs},
  description={discrete-time control barrier function},
  first={discrete-time control barrier function (\glsentrytext{DCBF})},
  firstplural={discrete-time CBFs (\glsentryplural{DCBF})}
}

\newglossaryentry{CBF}
{
  name={CBF},
  plural={CBFs},
  description={control barrier function},
  first={control barrier function (\glsentrytext{CBF})},
  firstplural={control barrier functions (\glsentryplural{CBF})}
}

\newglossaryentry{Q-CBF}
{
  name={Q-CBF},
  description={state--action control barrier function},
  first={state--action control barrier function (\glsentrytext{Q-CBF})}
}

\newglossaryentry{ODD}
{
  name={ODD},
  description={Operational Design Domain},
  first={operational design domain (\glsentrytext{ODD})}
}

\newglossaryentry{LRSF}
{
  name={LRSF},
  description={Least-Restrictive Safety Filter},
  first={least-restrictive safety filter (\glsentrytext{LRSF})}
}

\newglossaryentry{HCSF}
{
  name={HCSF},
  description={Human-Centered Safety Filter},
  first={human-centered safety filter (\glsentrytext{HCSF})}
}

\newglossaryentry{NLP}
{
  name={NLP},
  description={nonlinear programming problem},
  first={nonlinear programming problem (\glsentrytext{NLP})},
}

\newglossaryentry{ILQR}
{
  name={ILQR},
  description={iterative linear quadratic regulator},
  first={iterative linear quadratic regulator (\glsentrytext{ILQR})},
}

\newglossaryentry{MPC}
{
  name={MPC},
  description={model predictive control},
  first={model predictive control (\glsentrytext{MPC})},
}

\newglossaryentry{AC}
{
  name={AC},
  description={Assetto Corsa},
  first={Assetto Corsa (\glsentrytext{AC})},
}

\newglossaryentry{SAC}
{
  name={SAC},
  description={Soft Actor--Critic},
  first={soft actor--critic (\glsentrytext{SAC})},
}

\newglossaryentry{ANOVA}
{
    name={ANOVA},
    description={Analysis of Variance},
    first={analysis of variance (\glsentrytext{ANOVA})}
}

\newglossaryentry{SME}
{
    name={SME},
    description={Simple Main Effects},
    first={simple main effects (\glsentrytext{SME})},
}

\newglossaryentry{HSD}
{
    name={HSD},
    description={Honestly Significant Difference},
    first={honestly significant difference (\glsentrytext{HSD})},
}

\newglossaryentry{ECDF}
{
    name={ECDF},
    description={Empirical Cumulative Distribution Function},
    first={empirical cumulative distribution function (\glsentrytext{ECDF})},
}

\newglossaryentry{OCP}
{
    name={OCP},
    description={Optimal Control Problem},
    first={optimal control problem (\glsentrytext{OCP})}
}

\newglossaryentry{AI}
{
    name={AI},
    description={Artificial Intelligence},
    first={artificial intelligence (\glsentrytext{AI})}
}

\newglossaryentry{HRI}
{
    name={HRI},
    description={Human--Robot Interaction},
    first={human--robot interaction (\glsentrytext{HRI})}
}

\newglossaryentry{MDP}
{
    name={MDP},
    description={Markov Decision Process},
    first={Markov decision process (\glsentrytext{MDP})}
}

\newglossaryentry{SC-MDP}
{
    name={SC-MDP},
    description={Safety-Critical Markov Decision Process},
    first={safety-critical Markov decision process (\glsentrytext{SC-MDP})}
}

\newacronym[longplural={constrained Markov decision processes}]{CMDP}{CMDP}{constrained Markov decision process}

\newglossaryentry{DP}
{
    name={DP},
    description={Dynamic Programming},
    first={dynamic programming (\glsentrytext{DP})}
}

\newglossaryentry{CPO}
{
    name={CPO},
    description={Constrained Policy Optimization},
    first={Constrained Policy Optimization (\glsentrytext{CPO})}
}

\newglossaryentry{RCPO}
{
    name={RCPO},
    description={Reward Constrained Policy Optimization},
    first={Reward Constrained Policy Optimization (\glsentrytext{RCPO})}
}

\newacronym{MPSF}{MPSF}{model predictive safety filter}

\newacronym{ISAACS}{ISAACS}{Iterative Soft Adversarial Actor–Critic for Safety}

\newacronym{QP}{QP}{quadratic program}

\newacronym[longplural={Gaussian processes}]{GP}{GP}{Gaussian process}

\newacronym{CVaR}{CVaR}{conditional value-at-risk}

\title{\LARGE \bf
Synthesis and Deployment of Maximal Robust Control Barrier Functions through Adversarial Reinforcement Learning
}

\author{Donggeon David Oh, Duy P. Nguyen, Haimin Hu, and Jaime Fernández Fisac%
\thanks{This work has been partially supported by the DARPA Learning Introspective Control (LINC) program (award number N6523623C8002) and by the NSF CAREER Award (number 2340851).}%
\thanks{D. D. Oh, D. P. Nguyen, and J. F. Fisac are with the Department of Electrical and Computer Engineering, Princeton University, Princeton, NJ 08540, USA 
{\tt\small \{do9948, duyn, jfisac\}@princeton.edu}
}%
\thanks{H. Hu is with the Department of Computer Science, Johns Hopkins University, Baltimore, MD 21218, USA
{\tt\small haimin@cs.jhu.edu}}%
}

\begin{document}

\makeatletter
\let\@oldmaketitle\@maketitle
\renewcommand{\@maketitle}{\@oldmaketitle%
  \centering
  \includegraphics[width=\textwidth]{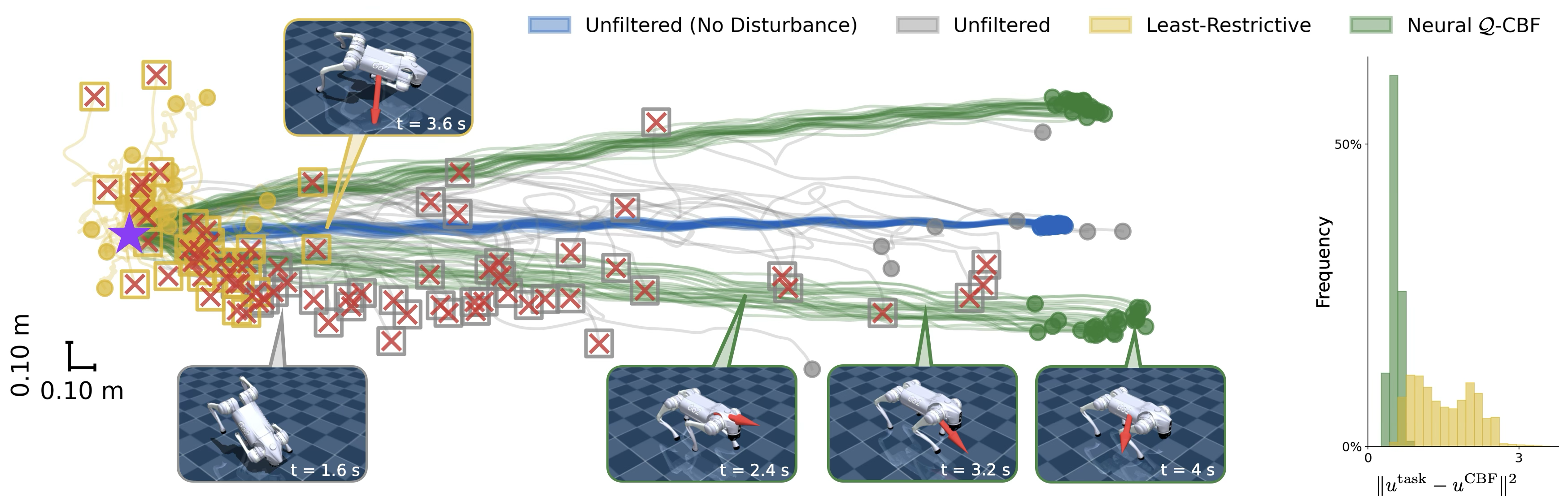}
  \setcounter{figure}{0}
  \vspace{-0.2in}
  \captionof{figure}{
    Our proposed neural $\qfunc$-CBF yields reliable safety enforcement on a simulated 36-D quadruped with \emph{black-box dynamics} under adversarial uncertainty realizations. The robot uses a pure-pursuit task policy to move right from the purple starting point. Over 50 trials, the safe rates are 100\% for $\qfunc$-CBF, 38\% for the least-restrictive safety filter (LRSF), and 16\% for the unfiltered policy. \emph{Left:} Trajectory comparison of the robot under the learned best-response disturbance policy. Red crosses indicate failures. $\qfunc$-CBF (green) enables stable forward locomotion, while the LRSF baseline (yellow) often prevents meaningful forward progress. \emph{Right:} The histogram of task input deviation $\left\|\ctrl^\task-\ctrl^\text{CBF}\right\|^2$ shows that the neural $\qfunc$-CBF enforces substantially smaller task input modifications than LRSF, leading to better preservation of task performance.
  }
  \label{fig:go2_exp_result}
  \vspace{-0.2in}
  \bigskip}
\makeatother

\maketitle
\thispagestyle{empty}
\pagestyle{empty}

\begin{abstract}
Robust control barrier functions (CBFs) provide a principled mechanism for smooth safety enforcement under worst-case disturbances. 
However, existing approaches typically rely on explicit, closed-form structure in the dynamics (e.g.,
control-affine) and uncertainty models.
This has led to limited scalability and generality, with most robust CBFs certifying only conservative subsets of the maximal robust safe set. 
In this paper, we introduce a new robust CBF framework for general nonlinear systems under bounded uncertainty. 
We first show that the safety value function solving the dynamic programming Isaacs equation is a valid robust discrete-time CBF that enforces safety on the maximal robust safe set. 
We then adopt the key reinforcement learning (RL) notion of \textit{quality function} (or Q-function), which
removes the need for explicit dynamics by lifting the barrier certificate into state--action space and yields a novel robust $\qfunc$-CBF constraint for safety filtering.
Combined with adversarial RL, this enables the synthesis and deployment of robust $\qfunc$-CBFs on general nonlinear systems with black-box dynamics and unknown uncertainty structure.
We validate the framework on a canonical inverted pendulum benchmark and a 36-D quadruped simulator, achieving substantially less conservative safe sets than barrier-based baselines on the pendulum and reliable safety enforcement even under adversarial uncertainty realizations on the quadruped.
\end{abstract}

\section{Introduction}
\label{sec:Intro}
Safety-critical systems are increasingly deployed in real-world environments, where uncertainty is unavoidable and even a single safety violation may have catastrophic consequences. This calls for a \emph{robust safety filter}, a runtime process that monitors the system’s operation and intervenes, when necessary, by modifying the control input to preserve safety against all admissible uncertainty realizations~\cite{hsu2024safety}.

Among robust safety filters, robust \glspl{CBF} are often sought after because they enable smooth intervention by solving a safety-enforcing \gls{OCP} at each timestep. Despite their popularity, however, current robust \gls{CBF} approaches face important challenges. Most robust \gls{CBF} schemes require \emph{explicit knowledge} of the control-affine dynamical model for offline synthesis, runtime constraint evaluation, or, in many cases, both~\cite{jankovic2018robust, lopez2020robust, alan2023parameterized}. They also rely on domain-specific structural assumptions of the uncertainty to bound or tractably compute its effect on the dynamics or on the barrier derivative~\cite{alan2022disturbance, cohen2022robust, buch2021robust, emam2021data}. Data-driven extensions alleviate some modeling burden, but still assume substantial model structure, such as closed-form nominal models and known error bounds~\cite{taylor2021towards,lindemann2024learning}. As a result, synthesizing and deploying robust \glspl{CBF} for systems with complex, possibly black-box dynamics and uncertainty structure remains an outstanding challenge. Moreover, depending on how uncertainty is handled, these methods often become \emph{overly conservative} and certify only subsets of the maximal robust safe set.

\gls{HJI} reachability analysis offers a complementary route to robust safety by formulating safety under bounded disturbance, which represents predictive uncertainty such as unknown model error and external perturbations, as a zero-sum game whose value function encodes the maximal robust safe set~\cite{mitchell2005time, chen2018hamilton, bansal2017hamilton}. Although classical \gls{HJI} methods scale poorly with state dimension, recent \gls{RL}-based methods have enabled scalable approximations of safety value functions and fallback policies, yielding safety filters with switch-type interventions for high-dimensional black-box systems under uncertainty~\cite{hsu2023isaacs,wang2024magics, nguyen2025gameplay}.

Recent work connects reachability and \glspl{CBF} by showing that exact reachability-derived values are valid \glspl{CBF} in a generalized sense~\cite{choi2021robust, hirsch2025viscosity, begzadic2025back}.
However, evaluating the corresponding \gls{CBF} constraints at runtime requires known dynamics and value derivatives, with control- and disturbance-affine dynamics allowing a \gls{QP} formulation.
While learned safety values have been used as \glspl{DCBF} on black-box systems~\cite{oh2025safety}, that formulation only considers disturbance-free systems.

In this paper, we draw a key insight from \gls{RL}---the quality function~\cite{watkins1992q}---that lifts the safety value function into a state--control--disturbance representation, enabling \gls{CBF}-style safety filtering without explicit dynamics or uncertainty models. We consider systems available only as \emph{black-box} transition mechanisms: given the current state, a control input, and a bounded disturbance input, the simulator or physical system returns the next state. Within this setting, we show that the safety value function itself is a valid robust \gls{DCBF}, and that its lift yields a novel robust $\qfunc$-\gls{CBF} constraint for safety filtering on \emph{black-box systems with unknown dynamics and uncertainty structure}.

The main contributions of this paper are as follows:
\begin{itemize}
    \item We introduce the robust $\qfunc$-\gls{CBF} framework for \emph{black-box} nonlinear systems under bounded uncertainty. 
    We formally show that the safety value function---solution to the dynamic programming Isaacs equation---is a valid robust \gls{DCBF} on the \emph{maximal} robust safe set, and that its state--action lift yields a novel robust \gls{CBF} constraint for smooth safety filtering under uncertainty.
    
    \item Leveraging reachability-based adversarial \gls{RL}, we develop a scalable robust $\qfunc$-\gls{CBF} synthesis and deployment pipeline for high-dimensional systems without requiring explicit dynamics, control-affine assumptions, prescribed uncertainty structure, or manual barrier function design.

    \item We validate our robust $\qfunc$-\gls{CBF} on a disturbed inverted pendulum (Fig.~\ref{fig:invpen_exp_result}) and a high-fidelity quadrupedal locomotion simulation governed by 36-dimensional black-box dynamics (Fig.~\ref{fig:go2_exp_result}).
    On the pendulum, the learned robust $\qfunc$-\gls{CBF} nearly recovers the maximal robust safe set and is substantially less conservative than barrier-based baselines. On the quadruped, it reliably enforces safety under adversarial uncertainty realizations while achieving effective forward locomotion.
\end{itemize}

\section{Preliminaries and Related Work}

\subsection{Robust Safety under Bounded Uncertainty}
\label{subsec:Robust_Safety_under_Bounded_Uncertainty}
Throughout this work, we consider a discrete-time system with nonlinear dynamics
\begin{equation}
    \label{eq:dynamics}
    \state_{t+1}=\dyn(\state_t, \ctrl_t, \dstb_t),
\end{equation}
where at each time $t$, $\state_t\in\xset$ is the system's state, $\ctrl_t\in\cset$ is the control input, and $\dstb_t\in\dset$ is an unknown but bounded disturbance input representing predictive uncertainty. 
We further consider a \textit{failure set} $\failureset\subset\xset$ of states that are deemed unacceptable and must be categorically avoided at all times.
The geometry of $\failureset$ may be arbitrary; by convention, we assume only that it is an open set and can therefore be characterized through some continuous margin function $\margin:\xset\rightarrow\mathbb{R}$ (e.g., a signed distance function) as
\begin{equation}
    \failureset\coloneqq\{\state\mid \margin(\state)<0\}.
\end{equation}
These assumptions can be viewed as specifying an \emph{\gls{ODD}}, namely the set of operating conditions under which the system is expected to function correctly and safely~\cite{hsu2024safety, on2021taxonomy}. In particular, the bounded disturbance set $\dset$, together with the modeling assumptions in \eqref{eq:dynamics}, delineates the operating conditions under which the guarantees are intended to hold.

Robust safety analysis aims to determine the \emph{maximal robust safe set} $\maxsafeset\subset\xset$, consisting of all states from which there exists a control policy guaranteed to keep the system outside of~$\failureset$ for all time under \textit{all possible realizations} of $\dstb\in\dset$:
\begin{align}
    \maxsafeset=\bigl\{ \state\in\xset\,\mid\,& \exists\cpolicy:\xset\rightarrow \cset,\,\forall\dpolicy:\xset\rightarrow\dset,\,\forall t\ge 0,\nonumber\\
    &\mathrm{\state}^{\cpolicy,\dpolicy}_\state(t)\notin\failureset\bigr\},
\end{align}
where $\mathrm{\state}^{\cpolicy,\dpolicy}_\state:\mathbb{N}\rightarrow\xset$ represents the state trajectory of the system~\eqref{eq:dynamics} starting from $\state_0=\state$ under the control policy $\cpolicy$ and the disturbance policy $\dpolicy$.

A widely used approach to enforcing safety under uncertainty is robust \gls{CBF}, which certifies forward invariance of a (usually conservative) subset of the maximal robust safe set.

\subsection{Robust Continuous-Time Control Barrier Function}
\label{subsec:RCTCBF}

The existing robust-\gls{CBF} literature is predominantly formulated in continuous time, so we begin by recalling the robust continuous-time \gls{CBF} before returning to the discrete-time setting of primary interest in \autoref{subsec:RDTCBF}.

\begin{definition}[Robust Continuous-Time CBF]
\label{def:RCTCBF}
Consider the continuous-time uncertain system
\begin{equation}
\label{eq:ct_dynamics}
\dot{\state}=\dyn_c(\state,\ctrl,\dstb),
\end{equation}
where $\state\in\xset$, $\ctrl\in\cset$, and $\dstb\in\dset$. A function $\cbf:\xset\rightarrow\reals$ is a robust \gls{CBF} for \eqref{eq:ct_dynamics} if $\safeset=\{\state\in\xset\mid \cbf(\state)\ge 0\}\subset\failureset^\compl$ and there exists a class-$\mathcal{K}$ function $\alpha$, defined on a domain containing $\cbf(\safeset)$, such that
\begin{equation}
\label{eq:RCTCBF_def}
\sup_{\ctrl\in\cset}\inf_{\dstb\in\dset}\ \dot{\cbf}(\state,\ctrl,\dstb)\ge -\alpha(\cbf(\state)),
\qquad\forall \state\in\safeset.
\end{equation}
\end{definition}

\begin{remark}[Choice of Function $\alpha$ in \autoref{def:RCTCBF}]
Since \eqref{eq:RCTCBF_def} is imposed only on $\safeset$, $\alpha$ may be taken from the standard class-$\mathcal{K}$ family. Extended class-$\mathcal{K}$ or extended class-$\mathcal{K}_\infty$ functions are needed when the barrier condition is imposed globally on $\xset$, so that negative values of $\cbf$ are also allowed, while class-$\mathcal{K}_\infty$ or extended class-$\mathcal{K}_\infty$ choices are natural when the range of $\cbf$ is unbounded. Such global formulations may additionally encode set-attractiveness properties, but this loses significance when $\cbf$ encodes $\maxsafeset$, since outside $\maxsafeset$ there exists no control policy that can guarantee safety against all admissible disturbances.
\end{remark}

Now, we explain why existing robust \gls{CBF} formulations can be viewed as special cases of \autoref{def:RCTCBF}.

Existing robust \gls{CBF} formulations typically assume an explicit control-affine model with structured uncertainty to either validate a robust \gls{CBF} or evaluate the robust \gls{CBF} constraint online, which we summarize here:
\begin{equation}
\label{eq:lit_uncertain_affine}
\dot{\state}
=
f(\state)+g(\state)\ctrl+\mu(\state,\ctrl).
\end{equation}
In many of these works, the required closed-form structure appears directly through the computation of $\dot{\cbf}$, for example via Lie-derivative terms such as $L_f h(\state)$ and $L_g h(\state)$, together with a specific description of how uncertainty affects the dynamics or the barrier derivative. The model~\eqref{eq:lit_uncertain_affine} captures many of the uncertainty descriptions used across the robust \gls{CBF} literature, including additive or multiplicative disturbances, structured parametric uncertainty, and sector-bounded input uncertainty~\cite{jankovic2018robust, lopez2020robust, cohen2022robust, buch2021robust, emam2021data}. These structured system classes can therefore be viewed as special cases of the more general uncertain dynamics~\eqref{eq:ct_dynamics}. In our framework, by contrast, we assume neither control-affine structure, explicit knowledge of the dynamics, nor prescribed uncertainty structure for \gls{CBF} synthesis or runtime evaluation.

One major line of work handles uncertainty by incorporating a compensation term $\sigma$ into the \gls{CBF} constraint, for example in the form
\begin{equation}
\label{eq:lit_sigma_form}
L_f h(\state)+L_g h(\state)\ctrl-\sigma(\state,\ctrl)\ge -\alpha(h(\state)).
\end{equation}
Constructing $\sigma$ requires sufficient prior knowledge of the uncertainty structure, for example worst-case disturbance bounds, parameter estimation error bounds, or observer error bounds, to upper-bound the contribution of uncertainty to $\dot h$~\cite{jankovic2018robust, lopez2020robust, alan2023parameterized, alan2022disturbance}. Such methods fit naturally within \autoref{def:RCTCBF}, but they can be conservative because they upper-bound the worst-case effect of uncertainty rather than handling it explicitly. As a result, the robust safe set $\safeset$ is often a strict subset of the maximal robust safe set $\maxsafeset$.

A second line of work keeps the worst-case effect of uncertainty explicit by retaining the inner minimization over $\dstb\in\dset$ in \eqref{eq:RCTCBF_def}, rather than introducing a compensation term. Representative examples include duality-based reformulations of the inner minimization, sector-bounded input uncertainty, and convex hull disturbance models~\cite{cohen2022robust, buch2021robust, emam2021data}. Because they preserve the state- and input-dependent worst-case effect of uncertainty, these approaches may be less conservative than compensation-based methods, but they still require a tractable uncertainty representation and sufficient explicit knowledge of the system to validate a robust \gls{CBF} and evaluate the inner minimization. A particularly relevant reachability-based variant is the robust control barrier-value function framework~\cite{choi2021robust}, which recovers the maximal robust safe set by constructing the certificate as the viscosity solution of an \gls{HJI} variational inequality. Our robust $\qfunc$-\gls{CBF} framework improves on this work by removing the need for runtime deployment to rely on control- and disturbance-affine dynamics with known closed-form expressions.

Data-driven approaches have also been adopted for robust barrier or certificate synthesis. Lindemann et al.~\cite{lindemann2024learning} learn robust output \glspl{CBF} from safe expert demonstrations, and Taylor et al.~\cite{taylor2021towards} develop a broader data-driven robust control synthesis framework based on control certificate functions under actuation uncertainty. Despite leveraging data, these methods still assume substantial model structure, such as known nominal control-affine system models and model error bounds, which are then used to validate the certificates.

Taken together, the robust continuous-time \gls{CBF} formulations reviewed above can all be viewed as instances of \autoref{def:RCTCBF}, but only under stronger structural assumptions. To synthesize a robust \gls{CBF} or evaluate the corresponding constraint online, these methods typically rely on explicit knowledge of the control-affine system dynamics, together with prescribed structure for how uncertainty affects the dynamics or the barrier derivative. Moreover, they are often more conservative than necessary, certifying only subsets of the maximal robust safe set. By contrast, as shown in \autoref{sec:rqcbf} and \autoref{sec:Deployment_Robust_CBF_SF}, our robust $\qfunc$-\gls{CBF} framework treats the uncertain system as a \emph{black-box}. For both synthesis and runtime evaluation, it requires neither control-affine assumptions, explicit knowledge of the dynamics, nor prescribed uncertainty structure, while still recovering the maximal robust safe set.

\subsection{Robust Discrete-Time Control Barrier Function}
\label{subsec:RDTCBF}
We now make explicit how the robust continuous-time \gls{CBF} in \autoref{def:RCTCBF} relates to the discrete-time system of interest \eqref{eq:dynamics}. If a class-$\mathcal{K}$ function $\alpha$ is locally Lipschitz, then it can be associated with a class-$\mathcal{KL}$ function $\beta_\alpha$ by defining, for each $r\ge 0$, the map $\beta_\alpha(r,\cdot)$ as the solution of the initial value problem 
$\dot{y}=-\alpha(y),\; y(0)=r$
~\cite{hirsch2025viscosity}.
Recall that a class-$\mathcal{KL}$ function is class-$\mathcal{K}$ in its first argument and, for each fixed $r\ge 0$, is nonincreasing in its second argument and converges to $0$ as $t\to\infty$. Under sampling with period $\Delta t$, the continuous-time robust \gls{CBF} condition \eqref{eq:RCTCBF_def} induces the one-step discrete-time condition
$\cbf(\state_{t+1})\ge \beta_\alpha(\cbf(\state_t),\Delta t)$.
For a fixed sampling period $\Delta t$, the map $r\mapsto \beta_\alpha(r,\Delta t)$ is itself class-$\mathcal{K}$, and with a slight abuse of notation we denote this induced one-step map simply by $\beta(r)$.

Motivated by this continuous--discrete connection, we now introduce the robust \gls{DCBF} for the discrete-time uncertain dynamics \eqref{eq:dynamics}.

\begin{definition}[Robust Discrete-Time CBF]
\label{def:RDCBF}
A function $\cbf:\xset\rightarrow\reals$ is a robust \gls{DCBF} for system \eqref{eq:dynamics} if $\safeset=\{\state\in\xset\mid \cbf(\state)\ge 0\}\subset\failureset^\compl$ and there exists a class-$\mathcal{K}$ function $\beta$, defined on a domain containing $\cbf(\safeset)$, such that
\begin{equation}
\label{eq:RDCBF_def}
\sup_{\ctrl\in\cset}\inf_{\dstb\in\dset}\ \cbf\bigl(\dyn(\state,\ctrl,\dstb)\bigr)\ge \beta(\cbf(\state)),
\qquad\forall \state\in\safeset.
\end{equation}
\end{definition}

For any current state $\state\in\safeset$, safety can be enforced for system~\eqref{eq:dynamics} by solving, at each timestep, an \gls{OCP} that finds the safety-preserving action closest to the task input $\ctrl^\task$:
\begin{subequations}
\label{eq:robust_nlp}
\begin{align}
\ctrl(\state)=\argminB_{\ctrl\in\cset}&\quad\left\|\ctrl^\task-\ctrl\right\|^2, & \label{eq.rnlp_cost}\\
\st&\quad \inf_{\dstb\in\dset}\cbf\bigl(\dyn(\state,\ctrl,\dstb)\bigr)\ge \beta(\cbf(\state)),
\label{eq:RDCBF_constraint}
\end{align}
\end{subequations}
where \eqref{eq:RDCBF_constraint} is the \emph{robust \gls{DCBF} constraint}. 
As with its continuous-time counterpart, $\safeset$ need not coincide with the maximal robust safe set $\maxsafeset$; in general, it may be a strict subset. This conservativeness gap motivates the reachability-based synthesis of safety certificates reviewed next.

\subsection{Hamilton--Jacobi--Isaacs Reachability Analysis}
\label{subsec:HJIRA}
To close the conservativeness gap discussed above, we turn to Hamilton--Jacobi--Isaacs (HJI) reachability analysis. In the discrete-time system~\eqref{eq:dynamics} considered here, the maximal robust safe set $\maxsafeset$ can be computed by solving the dynamic-programming \emph{Isaacs equation}, which formulates the robust safety problem as a zero-sum \textit{safety game} between the controller and the disturbance~\cite{mitchell2005time, hsu2023isaacs}:
\begin{equation}
    \label{eq:HJI}
    \valfunc(\state_t)=\min\Bigl\{\margin(\state_t), \max_{\ctrl_t\in\cset}\min_{\dstb_t\in\dset}\valfunc\bigl(\dyn(\state_t, \ctrl_t, \dstb_t)\bigr)\Bigr\},
\end{equation}
where the \emph{safety value function} $\valfunc:\xset\rightarrow\mathbb{R}$ encodes the minimal margin that the controller can maintain at all times under the worst-case disturbance. Notably, the inner $\max_{\ctrl}\min_{\dstb}$ ordering grants the disturbance an instantaneous \emph{informational advantage} by allowing it to react to the controller’s chosen input at each timestep. Consequently, the maximal robust safe set is given by
\begin{equation}
    \label{eq:maximal_robust_safe_set}
    \maxsafeset=\bigl\{\state\in\xset\mid \valfunc(\state)\geq 0\bigr\}.
\end{equation}

We next adopt a key insight from \gls{RL}: the quality function (action-value function)~\cite{watkins1992q}. Specifically, we define the \emph{state--control--disturbance safety value function}
$\qfunc:\xset\times\cset\times\dset\rightarrow\mathbb{R}$ to satisfy
\begin{equation}
    \label{eq:Q-HJI}
    \qfunc(\state_t,\ctrl_t,\dstb_t)=\min\Bigl\{\margin(\state_t),\valfunc\bigl(\dyn(\state_t,\ctrl_t,\dstb_t)\bigr)\Bigr\}.
\end{equation}
This formulation remains equivalent to \eqref{eq:HJI} in the sense that
\begin{equation}
    \valfunc(\state)=\max_{\ctrl\in\cset}\min_{\dstb\in\dset}\qfunc(\state,\ctrl,\dstb).
\end{equation}
The introduction of $\qfunc$ lifts the safety certificate $\valfunc$ into a state--control--disturbance representation that directly underlies the robust $\qfunc$-\gls{CBF} constraint developed later.

We will also make use of the corresponding \emph{safe fallback policy}
\begin{equation}
    \label{eq:fallback}
    \fallback(\state)\coloneqq\argmaxB_{\ctrl\in\cset}\min_{\dstb\in\dset}\qfunc(\state,\ctrl,\dstb),
\end{equation}
which selects a control input maximizing the worst-case safety value. In conjunction with any task input $\ctrl^\task\in\cset$, the safety value function $\valfunc$ and fallback policy $\fallback$ induce a \gls{LRSF} with value-based monitoring and switch-type intervention~\cite{chen2018hamilton, bansal2017hamilton, hsu2024safety, oh2026provably}:
\begin{equation}
    \label{eq:LRSF}
    \ctrl(\state)=
    \begin{cases}
        \ctrl^\task, & \min_{\dstb\in\dset}\valfunc\bigl(\dyn(\state,\ctrl^\task,\dstb)\bigr)\ge 0,\\
        \fallback(\state), & \text{otherwise}.
    \end{cases}
\end{equation}

\section{Maximal Robust $\qfunc$-CBF}
\label{sec:rqcbf}
In this section, we establish the main theoretical result underlying our robust $\qfunc$-\gls{CBF} framework. Our key insight is that the safety value function $\valfunc$ itself is a valid robust \gls{DCBF} whose $0$-superlevel set is the \emph{maximal robust safe set} $\maxsafeset$. It follows that, for every class-$\mathcal{K}$ function $\beta$ satisfying $\beta(r)\le r$ for all $r\ge 0$, robust safety on $\maxsafeset$ can be enforced through the state--control--disturbance safety value function $\qfunc$ via the constraint
\[
\min_{\dstb\in\dset}\qfunc(\state,\ctrl,\dstb)\ge \beta(\valfunc(\state)).
\]
Crucially, once $\valfunc$ and $\qfunc$ are available, the $\qfunc$-\gls{CBF} constraint can be evaluated from these value functions alone, without explicit closed-form dynamics, control-affine assumptions, or prescribed uncertainty structure. In this sense, the constraint admits \emph{black-box} evaluation. Practical considerations for scalably synthesizing $\valfunc$ and $\qfunc$, together with handling the minimization over $\dstb$ at runtime, are introduced in \autoref{sec:Deployment_Robust_CBF_SF}. We now formalize this result.

\begin{lemma}
\label{lem:robust_exist_Q_bound}
For all $\state\in\maxsafeset$ and for every class-$\mathcal{K}$ function $\beta$ satisfying $\beta(r)\le r$ for all $r\ge 0$, there exists an action $\ctrl\in\cset$ such that
\begin{equation}
\label{eq:robust_exist_Q_bound}
\min_{\dstb\in\dset} \qfunc(\state,\ctrl,\dstb)\;\ge\;\beta(\valfunc(\state)),
\end{equation}
where $\qfunc$ and $\valfunc$ satisfy the state--control--disturbance Isaacs equation~\eqref{eq:Q-HJI} and the Isaacs equation~\eqref{eq:HJI}, respectively.
\end{lemma}

\begin{proof}
By the definition of the safe fallback policy~\eqref{eq:fallback},
\[
\min_{\dstb\in\dset}\qfunc(\state,\fallback(\state),\dstb)=\valfunc(\state).
\]
Since $\state\in\maxsafeset$, we have $\valfunc(\state)\ge 0$. Hence, for any class-$\mathcal{K}$ function $\beta$ satisfying $\beta(r)\le r$ for all $r\ge 0$,
\[
\min_{\dstb\in\dset}\qfunc(\state,\fallback(\state),\dstb)
=\valfunc(\state)\ge \beta(\valfunc(\state)).
\]
Therefore, $\ctrl=\fallback(\state)$ satisfies \eqref{eq:robust_exist_Q_bound}.
\end{proof}

\begin{lemma}
\label{lem:robust_QCBF_equivalence}
For all $\state\in\maxsafeset$, all $\ctrl\in\cset$, and every class-$\mathcal{K}$ function $\beta$ satisfying $\beta(r)\le r$ for all $r\ge 0$, the inequality
\begin{equation}
\label{eq:robust_QCBF_ineq}
\min_{\dstb\in\dset} \qfunc(\state,\ctrl,\dstb)\;\ge\;\beta(\valfunc(\state))
\end{equation}
is equivalent to
\begin{equation}
\label{eq:robust_DCBF_step}
\inf_{\dstb\in\dset}\valfunc\bigl(\dyn(\state,\ctrl,\dstb)\bigr)\;\ge\;\beta(\valfunc(\state)).
\end{equation}
\end{lemma}

\begin{proof}
First, \eqref{eq:HJI} implies
\begin{equation}
\label{eq:V_leq_g}
\valfunc(\state)\le \margin(\state), \qquad \forall \state\in\xset.
\end{equation}

For the forward direction, if \eqref{eq:robust_QCBF_ineq} holds, then
\[
\qfunc(\state,\ctrl,\dstb)\ge \beta(\valfunc(\state)), \qquad \forall \dstb\in\dset.
\]
Using \eqref{eq:Q-HJI}, this becomes
\[
\min\{\margin(\state),\,\valfunc(\dyn(\state,\ctrl,\dstb))\}\ge \beta(\valfunc(\state)),
\qquad \forall \dstb\in\dset,
\]
which immediately implies
\[
\valfunc(\dyn(\state,\ctrl,\dstb))\ge \beta(\valfunc(\state)), \qquad \forall \dstb\in\dset.
\]
Taking the infimum over $\dstb\in\dset$ yields \eqref{eq:robust_DCBF_step}.

For the reverse direction, assume \eqref{eq:robust_DCBF_step}. Then
\[
\valfunc(\dyn(\state,\ctrl,\dstb))\ge \beta(\valfunc(\state)), \qquad \forall \dstb\in\dset.
\]
Since $\state\in\maxsafeset$, we have $\valfunc(\state)\ge 0$, and therefore \eqref{eq:V_leq_g} together with $\beta(r)\le r$ for all $r\ge 0$ gives
\[
\margin(\state)\ge \valfunc(\state)\ge \beta(\valfunc(\state)).
\]
Hence
\[
\min\{\margin(\state),\,\valfunc(\dyn(\state,\ctrl,\dstb))\}\ge \beta(\valfunc(\state)),
\qquad \forall \dstb\in\dset.
\]
Using \eqref{eq:Q-HJI}, this is exactly
\[
\qfunc(\state,\ctrl,\dstb)\ge \beta(\valfunc(\state)), \qquad \forall \dstb\in\dset,
\]
and taking the minimum over $\dstb\in\dset$ gives \eqref{eq:robust_QCBF_ineq}.
\end{proof}

We now prove that the safety value function $\valfunc$ is a valid robust \gls{DCBF} on the maximal safe set $\maxsafeset$.
Consequently, we call it the \emph{maximal robust $\qfunc$-\gls{CBF}}.
Building on this result, we then present the corresponding robust $\qfunc$-CBF constraint.

\begin{theorem}[Maximal Robust $\qfunc$-CBF]
\label{thm:robust_QCBF}
The safety value function $\valfunc:\xset\rightarrow\reals$, which satisfies the Isaacs equation~\eqref{eq:HJI}, is a valid robust \gls{DCBF} whose $0$-superlevel set is the maximal robust safe set $\maxsafeset$, in the sense of \autoref{def:RDCBF}. In particular, for any class-$\mathcal{K}$ function $\beta$ satisfying $\beta(r)\le r$ for all $r\ge 0$, the robust $\qfunc$-\gls{CBF} constraint, which is equivalent on $\maxsafeset$ to the robust \gls{DCBF} constraint~\eqref{eq:RDCBF_constraint} with $\cbf=\valfunc$, is given by
\begin{equation}
\label{eq:robust_QCBF_constraint}
\min_{\dstb_t\in\dset}\qfunc(\state_t,\ctrl_t,\dstb_t)\ge \beta(\valfunc(\state_t)),
\end{equation}
where $\qfunc$ satisfies the state--control--disturbance Isaacs equation~\eqref{eq:Q-HJI}.
\end{theorem}

\begin{proof}
Since $\maxsafeset=\{\state\in\xset\mid \valfunc(\state)\ge 0\}\subset\failureset^\compl$, we now show that $\valfunc$ satisfies the robust \gls{DCBF} condition in \autoref{def:RDCBF} on $\maxsafeset$.

Let $\beta$ be any class-$\mathcal{K}$ function satisfying $\beta(r)\le r$ for all $r\ge 0$. By \autoref{lem:robust_exist_Q_bound}, for every $\state\in\maxsafeset$ there exists $\ctrl\in\cset$ such that
\[
\min_{\dstb\in\dset}\qfunc(\state,\ctrl,\dstb)\ge \beta(\valfunc(\state)).
\]
By \autoref{lem:robust_QCBF_equivalence}, this is equivalent to
\[
\inf_{\dstb\in\dset}\valfunc(\dyn(\state,\ctrl,\dstb))\ge \beta(\valfunc(\state)).
\]
Hence, for this choice of $\beta$, the function $\valfunc$ satisfies \autoref{def:RDCBF} with $\cbf=\valfunc$ on $\maxsafeset$. Since $\beta$ was arbitrary, $\valfunc(\cdot)$ is a valid robust \gls{DCBF} on $\maxsafeset$ for every class-$\mathcal{K}$ function satisfying $\beta(r)\le r$ for all $r\ge 0$. The equivalent robust $\qfunc$-\gls{CBF} constraint is precisely \eqref{eq:robust_QCBF_constraint}.
\end{proof}

We now leverage \autoref{thm:robust_QCBF} to introduce a robust $\qfunc$-\gls{CBF} safety filter that guarantees safety on the maximal robust safe set $\maxsafeset$. For any state $\state_t\in\maxsafeset$, the filter solves, at each timestep, an \gls{OCP} to find the control input closest to the task input $\ctrl^\task_t$ while satisfying the robust $\qfunc$-\gls{CBF} constraint~\eqref{eq:robust_QCBF_constraint}. This allows the filter to enforce safety against all admissible disturbance realizations $\dstb_t\in\dset$.

\begin{definition}[Robust $\qfunc$-CBF Safety Filter]
\label{def:Robust_CBF_SF}
For any class-$\mathcal{K}$ function $\beta$ satisfying $\beta(r)\le r$ for all $r\ge 0$, the robust $\qfunc$-\gls{CBF} safety filter is defined by
\begin{subequations}
\label{eq:Robust_CBF_SF}
\begin{align}
\ctrl(\state_t)=\argminB_{\ctrl_t\in\cset}&\quad\left \| \ctrl^\task_t-\ctrl_t \right \|^2, & \label{eq:Robust_CBF_SF_cost}\\
\st&\quad\min_{\dstb_t\in\dset}\qfunc(\state_t, \ctrl_t, \dstb_t)\geq\beta(\valfunc(\state_t)),
\label{eq:Robust_CBF_SF_constraint}
\end{align}
\end{subequations}
for $\state_t\in\maxsafeset$, where $\beta$ specifies how much of the safety value function must be preserved over a single timestep.
\end{definition}

\begin{remark}[Uncertainty-Free Special Case]
In the absence of uncertainty, the dynamics reduce to $\state_{t+1}=\dyn(\state_t,\ctrl_t)$ and the minimization over $\dstb_t$ in \eqref{eq:Robust_CBF_SF_constraint} disappears, yielding
\begin{equation*}
\tag{\ref{eq:Robust_CBF_SF}c}
\label{eq:Robust_CBF_SF_constraint_nod}
\qfunc(\state_t,\ctrl_t)\ge \beta(\valfunc(\state_t)).
\end{equation*}
This recovers the $\qfunc$-\gls{CBF} constraint proposed in~\cite{oh2025safety}.
\end{remark}

We emphasize the key difference between the robust continuous-time \gls{CBF} constraint~\eqref{eq:RCTCBF_def}, the robust \gls{DCBF} constraint~\eqref{eq:RDCBF_constraint}, and the robust $\qfunc$-\gls{CBF} constraint~\eqref{eq:Robust_CBF_SF_constraint}. Both \eqref{eq:RCTCBF_def} and \eqref{eq:RDCBF_constraint} depend explicitly on the system dynamics, and therefore require explicit knowledge of the dynamics and uncertainty in order to evaluate the constraint. By contrast, the robust $\qfunc$-\gls{CBF} constraint~\eqref{eq:Robust_CBF_SF_constraint} is posed directly in terms of the safety value functions $\valfunc$ and $\qfunc$. Thus, once $\valfunc$ and $\qfunc$ are available, evaluating \eqref{eq:Robust_CBF_SF_constraint} no longer requires explicit closed-form dynamics, control-affine assumptions, or prescribed uncertainty structure. Together with the practical method introduced in \autoref{sec:Deployment_Robust_CBF_SF} for handling the inner minimization over $\dstb$ at runtime, this formulation enables safety enforcement for black-box dynamical systems.

The recursive feasibility of the robust $\qfunc$-\gls{CBF} safety filter on the maximal robust safe set $\maxsafeset$ is a direct consequence of the safety value function $\valfunc$ being a valid robust \gls{DCBF}, as stated in \autoref{thm:robust_QCBF}.

\begin{proposition}[Recursive Feasibility of Robust $\qfunc$-\gls{CBF} Safety Filter] \label{pro:recursive_feasibility}
The optimization problem in \eqref{eq:Robust_CBF_SF} is recursively feasible for any class-$\mathcal{K}$ function $\beta$ satisfying $\beta(r)\le r$ for all $r\ge 0$, given any initial state $\state \in \maxsafeset$.
\end{proposition}

\begin{proof}
Fix any $\state_t\in\maxsafeset$. By \autoref{thm:robust_QCBF}, there exists $\ctrl_t\in\cset$ satisfying \eqref{eq:Robust_CBF_SF_constraint}, so \eqref{eq:Robust_CBF_SF} is feasible at time $t$.

Now let $\ctrl_t$ be any feasible solution of \eqref{eq:Robust_CBF_SF}. Then
\[
\qfunc(\state_t,\ctrl_t,\dstb_t)\ge \beta(\valfunc(\state_t)),
\qquad \forall \dstb_t\in\dset.
\]
By \eqref{eq:Q-HJI},
\[
\valfunc\bigl(\dyn(\state_t,\ctrl_t,\dstb_t)\bigr)\ge \qfunc(\state_t,\ctrl_t,\dstb_t),
\qquad \forall \dstb_t\in\dset,
\]
and therefore
\[
\valfunc\bigl(\dyn(\state_t,\ctrl_t,\dstb_t)\bigr)\ge \beta(\valfunc(\state_t)),
\qquad \forall \dstb_t\in\dset.
\]
Since $\state_t\in\maxsafeset$, we have $\valfunc(\state_t)\ge 0$, and because $\beta$ is class-$\mathcal{K}$, $\beta(\valfunc(\state_t))\ge 0$. Hence
\[
\valfunc\bigl(\dyn(\state_t,\ctrl_t,\dstb_t)\bigr)\ge 0,
\qquad \forall \dstb_t\in\dset,
\]
which implies $\state_{t+1}\in\maxsafeset$. Repeating the same argument at the next timestep proves recursive feasibility.
\end{proof}

\section{Synthesizing and Deploying Robust $\qfunc$-CBF
}
\label{sec:Deployment_Robust_CBF_SF}

The robust $\qfunc$-\gls{CBF} safety filter introduced in Section~\ref{sec:rqcbf} is computationally challenging for high-dimensional systems, both during synthesis and at deployment.
First, computing safety value function $\qfunc$ via direct solve of Isaacs equation \eqref{eq:HJI} is intractable due to the curse of dimensionality. 
In addition, deploying $\qfunc$-\gls{CBF} safety filter at runtime
requires optimizing disturbance input $\dstb$ when evaluating
\(
\min_{\dstb\in\dset} \qfunc(\state,\ctrl,\dstb)
\)
in the robust $\qfunc$-\gls{CBF} constraint \eqref{eq:Robust_CBF_SF_constraint}.
This leads to a \emph{nested} minimization over $\dstb\in\dset$, 
which is computationally intractable for high-dimensional disturbance spaces.
In this section, we introduce scalable robust $\qfunc$-\gls{CBF} synthesis and deployment procedures based on recently proposed game-theoretic adversarial \gls{RL}~\cite{hsu2023isaacs,wang2024magics}.

\subsection{Synthesizing Robust $\qfunc$-CBF via Game-theoretic RL}

\label{subsec:robust_neural_synthesis}
Our robust $\qfunc$-\gls{CBF} safety filter design assumes access to the state--control--disturbance safety value function $\qfunc$.
Instead of directly solving the Isaacs equation \eqref{eq:HJI}, we propose to synthesize it through an adversarial \gls{RL} process~\cite{hsu2023isaacs,wang2024magics}, where the controller and disturbance play a zero-sum dynamic game.

\p{Game-theoretic reinforcement learning}
The adversarial \gls{RL} jointly trains a critic (safety value) $\qfunc_\omega(\state,\ctrl,\dstb)$, controller actor $\cpolicy_\theta(\state)$, and disturbance actor $\dpolicy_\psi(\state,\ctrl)$.
Here, we allow the disturbance to observe and react to the control input, making it a stronger and more targeted adversary.
The critic minimizes the Bellman residual of the time-discounted Isaacs equation~\eqref{eq:HJI}:
\begin{equation*}
\begin{aligned}
L(\omega,&\theta,\psi)\;\coloneqq\; \expectation_{\xi \sim \buffer} \Bigl[\Bigl( \qfunc_{\omega}\bigl(\state,\ctrl,\dstb\bigr) - \Bigr.\Bigr.\\
& \Bigl.\Bigr. (1-\gamma_{\mathrm{ENV}})\,g' -\gamma_{\mathrm{ENV}}\min\Bigl\{g', \qfunc_{\omega'}\bigl(\state',\ctrl',\dstb'\bigr)\Bigr\}\Bigr)^2 \Bigr],
\end{aligned}
\end{equation*}
where $\xi = (\state,\ctrl,\dstb,g',\state')$ is a transition sampled from the replay buffer $\buffer$,
$\ctrl'\sim \cpolicy_\theta(\cdot\mid\state')$, $\dstb'\sim \dpolicy_\psi(\cdot\mid\state',\ctrl')$, 
$g':=\margin(\state')$,
$\gamma_{\mathrm{ENV}}\in(0,1)$  is the discount factor,
and $\qfunc_{\omega'}$ is a slowly updated target critic.
Following standard adversarial \gls{RL} practice, the controller maximizes critic $\qfunc_{\omega}$ (with an additional exploration term) while the disturbance minimizes the same critic.

\p{GDA with finite timescale separation}
In order to seek a minimax equilibrium of the safety game, we leverage
\gls{GDA} with \textit{finite timescale separation}~\cite{fiez2021local,jin2020local} between the learning rates of the controller and disturbance.
Specifically, the disturbance's policy parameters $\psi$ are updated on a faster timescale than the controller parameters $\theta$.
This strategy stabilizes learning by encouraging the disturbance actor to track a near-best response to the current controller actor.
Formally, it has been shown in Wang and Hu et al.~\cite{wang2024magics} that the two-timescale RL training guarantees convergence to a local minimax equilibrium $(\cpolicy_{\theta^*}, \dpolicy_{\psi^*})$ in the \textit{policy space}. 
At such an equilibrium, the learned disturbance policy represents a locally \textit{best-effort worst-case} response to the controller, yielding a principled approximation of the inner minimization $\min_{d\in D} \qfunc(x,u,d)$ in constraint \eqref{eq:Robust_CBF_SF_constraint} and enabling tractable computation of the robust $\qfunc$-CBF safety filtering at runtime.

\p{Training a best-response disturbance policy}
While the \gls{GDA}-trained disturbance policy $\dpolicy_{\psi^*}: \state \mapsto \dstb$ is a worst-case response to the \textit{equilibrium} control policy $\cpolicy_{\theta^*}: \state \mapsto \ctrl$, it does not necessarily minimize the critic $\qfunc_{\omega^*}$ for \textit{arbitrary} control inputs generated by $\qfunc$-CBF optimization~\eqref{eq:Robust_CBF_SF},
which may be exploitable.
To ensure local robustness of $\qfunc$-CBF, we further 
train
a \emph{best-response disturbance policy} $\dpolicy_{\tilde{\psi}}$ for a range of diverse control policies.
We do so by training the disturbance policy $\dpolicy_{\psi}$
by minimizing $\qfunc_{\omega}(\state,\ctrl,\dstb)$ with respect to $\dstb$, where $(\state,\ctrl)$ pairs are sampled from
a diverse set of control policies that play against the disturbance.
In practice, we collect these control policies from checkpoints of multiple RL training runs with different random seeds, thereby exposing the disturbance policy to a broad distribution of controller behaviors.
This procedure encourages $\dpolicy_{\tilde{\psi}}(\state,\ctrl)$ to approximate a local minimizer of the critic across the entire state--control space.

\subsection{Real-time Safety Filtering with Robust Neural $\qfunc$-CBF}
\label{subsec:QCBF_deployment}

At runtime, the neural-approximated disturbance policy $\dpolicy_{\tilde{\psi}}(\state,\ctrl)$ allows for efficient evaluation of $\min_{\dstb\in\dset}\qfunc(\state,\ctrl,\dstb)$ inside the robust $\qfunc$-\gls{CBF} constraint \eqref{eq:Robust_CBF_SF_constraint}. 
Specifically, we use the plug-in approximation $\tilde{\dstb} = \dpolicy_{\tilde{\psi}}(\state,\ctrl)$ to obtain $\qfunc(\state,\ctrl,\tilde{\dstb})$ as a tractable surrogate of the worst-case safety value.
This substitution removes the need for nested optimization in $\qfunc$-\gls{CBF} safety filtering \eqref{eq:Robust_CBF_SF} by reducing constraint evaluation \eqref{eq:Robust_CBF_SF_constraint} to a single forward pass through the disturbance policy network and safety critic.
Importantly, this plug-in choice is not arbitrary: 
$\dpolicy_{\tilde{\psi}}$ is trained as a best-response policy that minimizes the safety critic against arbitrary control policies within the same safety game defined by~\eqref{eq:HJI}.
Since $\dpolicy_{\tilde{\psi}}$ is locally worst-case, any nearby disturbance \textit{realization} would produce a less adversarial disturbance action at runtime.
As a result, enforcing the robust $\qfunc$-CBF constraint assuming the neural disturbance action $\tilde{\dstb}=\dpolicy_{\tilde{\psi}}(\state,\ctrl)$ automatically guarantees safety against all disturbance realizations that are sufficiently close.
This insight is formalized in the following proposition; the proof follows directly from local optimality of $\dpolicy_{\tilde{\psi}}(\state,\ctrl)$.

\begin{proposition}[Local Robustness of Neural $\qfunc$-CBF]
    \label{prop:local_robustness}

    Suppose $\dpolicy_{\tilde{\psi}}(\state,\ctrl)$ is a locally optimal disturbance policy, \ie,
    \begin{equation*}
    \begin{aligned}
    \qfunc\bigl(\state,\ctrl,\dpolicy_{\tilde{\psi}}(\state,\ctrl)\bigr)
    &\;\le\;
    \qfunc\bigl(\state,\ctrl,\dpolicy_{\psi}(\state,\ctrl)\bigr),
    \end{aligned}
    \end{equation*}
    for all $\state \in \xset$, $\ctrl \in \cset$, and $\psi \in B_\rho(\tilde{\psi})$ with $\rho>0$.
    If the surrogate constraint
    \[
    \qfunc(\state,\ctrl,\dpolicy_{\tilde{\psi}}(\state,\ctrl))
    \ge
    \beta(\valfunc(\state))
    \]
    holds for a given state $\state$ and control action $\ctrl$, then
    \[
    \qfunc(\state,\ctrl,\dpolicy_{\psi}(\state,\ctrl))
    \ge
    \beta(\valfunc(\state))
    \]
    also holds for all disturbance policies with $\psi \in B_\rho(\tilde{\psi})$.

\end{proposition}

\begin{remark}[Safety Verification]
While the use of function approximators in neural $\qfunc$-CBF safety filtering precludes direct reliance on the theoretical properties of the exact solution to~\eqref{eq:HJI}, we can inherit some of them in the form of explicit certificates for the approximate control policy and safety value function, obtained through recently developed post-hoc safety verification methods such as offline statistical calibration through conformal prediction~\cite{lin2024verification} and runtime verification using \textit{imagined} gameplay rollouts with faster-than-real-time, black-box simulators~\cite{hsu2023isaacs,nguyen2025gameplay}.
\end{remark}

\section{Experiments}
We empirically validate the proposed robust $\qfunc$-\gls{CBF} framework on a disturbed inverted pendulum and a simulated 36-dimensional quadruped with black-box dynamics. On the pendulum, the learned robust $\qfunc$-\gls{CBF} nearly recovers the maximal robust safe set. On the quadruped, it reliably enforces safety under adversarial uncertainty realizations while preserving forward locomotion.

\subsection{Disturbed Inverted Pendulum}
\label{subsec:exp_pendulum}

We consider the disturbed inverted pendulum benchmark of Alan et al.~\cite{alan2023parameterized}. Let $\theta\in[-\pi,\pi]$ denote the angle from the upright equilibrium and let $\omega=\dot{\theta}$ denote angular velocity. The dynamics are given by
\begin{equation}
\label{eq:exp_pendulum_dynamics}
\dot{\theta}=\omega,\qquad
\dot{\omega}=10\sin\theta+\frac{u}{2}+\frac{F}{2}\cos\theta,
\end{equation}
where $u\in[-20,20]$ and $F\in[-2,2]$. Unlike the original setup in~\cite{alan2023parameterized}, we impose the control bound explicitly. Failure occurs when $|\theta|>\pi/3$.

We compare the $0$-superlevel set of the learned robust $\qfunc$-\gls{CBF} against the following baselines:

\p{Heuristic barrier candidate}
Following Alan et al.~\cite{alan2023parameterized}, we use
$h_{\mathrm{heu}}(\theta,\omega)=1-16\theta^2-8\theta\omega-4\omega^2.$
Because the original setup does not impose the control bound, $h_{\mathrm{heu}}$ is not valid under our setting and is included only as a heuristic baseline.

\p{Analytically designed CBF}
We also include a robust \gls{CBF} that is valid in the sense of \autoref{def:RCTCBF}, derived using the nominal evading maneuver~\cite{oh2023safety}
$u^\star(\omega):=-20\,\sgn(\omega)$:
$h_{\mathrm{ana}}(\theta,\omega)=\cos\theta-\frac{1}{2}-\frac{\sqrt{3}}{2(19-10\sqrt{3})}\omega^2.$

\p{Maximal robust safe set}
We solve the Isaacs equation~\eqref{eq:HJI} on a grid using \texttt{OptimizedDP}~\cite{bui2022optimizeddp}.

Fig.~\ref{fig:invpen_exp_result} shows that the $0$-superlevel set of the neural $\qfunc$-\gls{CBF} is \emph{substantially less conservative} than the barrier-based baselines and \emph{nearly recovers the maximal robust safe set}. For all filters, we use the same PD task controller as in Alan et al.~\cite{alan2023parameterized}. To stress-test robustness, for each filter, we freeze the resulting executed control at each state on the grid and numerically compute the corresponding best-response disturbance against that closed-loop controller. We then initialize 20 trajectories near the boundary of its $0$-superlevel set and simulate them under these filter-specific worst-case disturbances. All three filtering schemes achieve 100\% empirical safety in this test.

\begin{figure}
    \centering
    \includegraphics[width=\columnwidth, trim=8cm 2cm 2cm 1cm, clip]{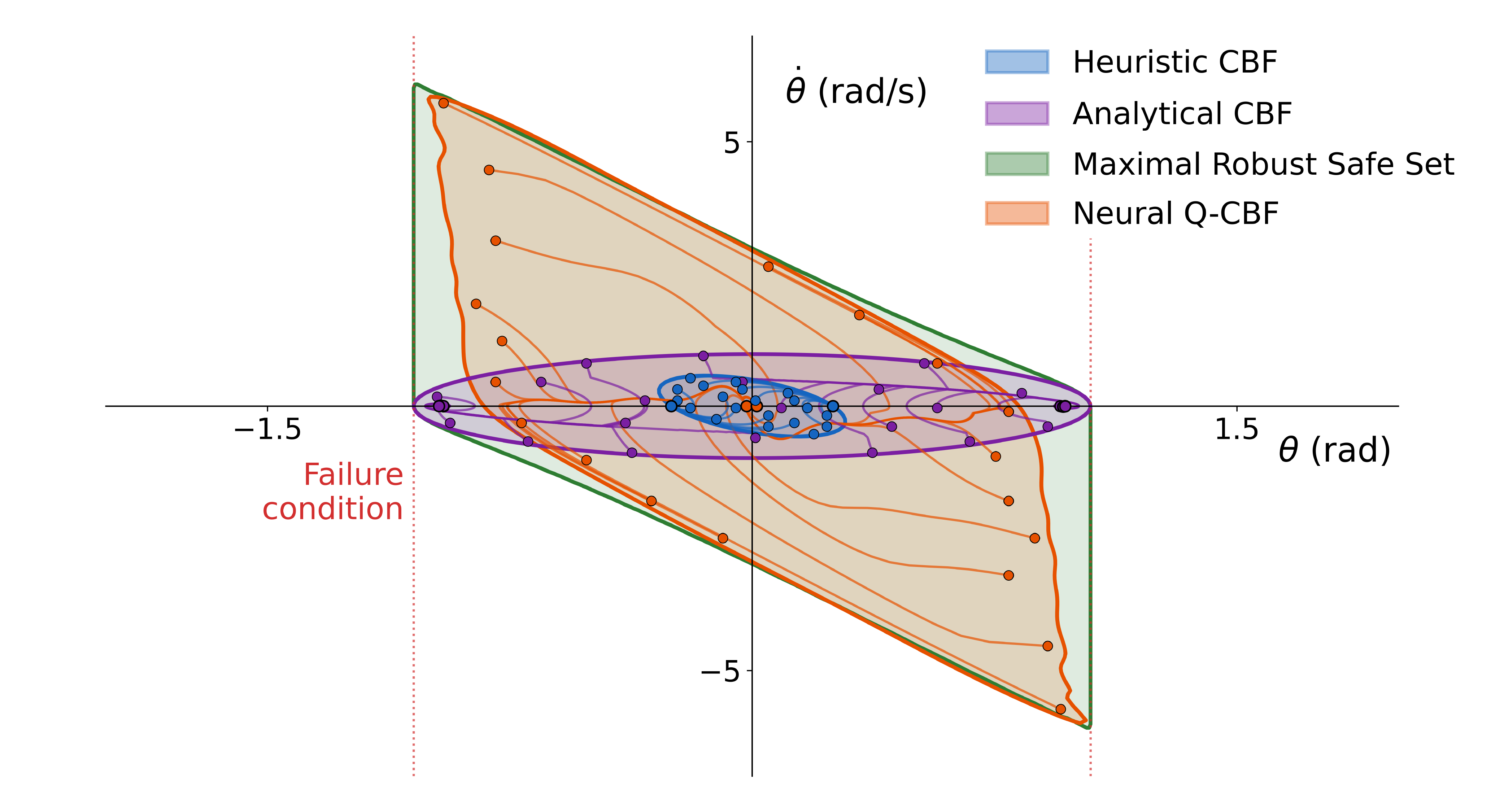}
    \caption{Safe sets and rollout trajectories for the disturbed inverted pendulum. The $0$-superlevel set of the neural $\qfunc$-CBF is substantially less conservative than the barrier-based baselines and nearly recovers the maximal robust safe set. For each method, 20 trajectories are initialized near the boundary of its $0$-superlevel set and simulated under the corresponding numerical best-response disturbance.}
    \label{fig:invpen_exp_result}
    \vspace{-5mm}
\end{figure}

\subsection{Quadrupedal Locomotion}
\label{subsec:exp_quadruped}

To validate that robust $\qfunc$-\gls{CBF} can reliably enforce safety on a high-dimensional black-box system under adversarial uncertainty realizations, we consider a high-fidelity simulation of a quadrupedal locomotion task with the 36-dimensional Unitree Go2 robot in MuJoCo~\cite{nguyen2025gameplay}. We treat the simulator as a black-box transition mechanism taking as inputs a 12-D control input, given by the joint-position-increment vector 
$\ctrl = [\delta\theta_{\mathrm J}^1,\ldots,\delta\theta_{\mathrm J}^{12}]^\top \in [-0.5,0.5]^{12}$,
and a bounded disturbance input consisting of an external force of magnitude at most $50\,\mathrm{N}$, whose application point on the torso and direction may be chosen arbitrarily. Failure is declared whenever any monitored torso-corner ground clearance falls below $0.10\,\mathrm{m}$.

To compare task performance with \gls{LRSF}~\eqref{eq:LRSF}, we first train a best-effort fallback policy together with its best-response disturbance through adversarial \gls{RL}. 
The resulting best-effort fallback policy is used for the \gls{LRSF} baseline. 
For the neural $\qfunc$-\gls{CBF}, we further train the best-response disturbance policy following \autoref{subsec:robust_neural_synthesis}. 
All methods are evaluated under this same disturbance policy that was trained adversarially against a spectrum of control policies.

Fig.~\ref{fig:go2_exp_result} compares the unfiltered task policy, \gls{LRSF}, and neural $\qfunc$-\gls{CBF}. The task policy is a pure-pursuit controller that drives the robot from the starting point (purple star) toward 
the right of the map. Under adversarial disturbance, the unfiltered 
task policy records only a 16\% safe rate. 
\gls{LRSF}'s loss in robustness due to neural approximation error is amplified by the last-minute intervention, yielding a low safe rate of 38\%.
Moreover, \gls{LRSF} produces abrupt chattering behaviors due to frequent switches between the task and fallback policies, a widely reported phenomenon~\cite{oh2025safety,hu2024active,bejarano2023multi}, preventing the robot from making meaningful forward progress.
By contrast, the neural $\qfunc$-\gls{CBF} preserves stable forward locomotion while successfully enforcing safety across all 50 randomized trials. The histogram of per-step task input deviation $\|\ctrl^\task-\ctrl^\text{CBF}\|^2$ is concentrated at substantially smaller values for neural $\qfunc$-\gls{CBF} than for \gls{LRSF}, suggesting that the neural $\qfunc$-\gls{CBF} better preserves task performance by enforcing smaller modifications to the task input than the baseline \gls{LRSF}.

\section{Conclusion}
\label{sec:conclusion}
We introduced a new robust \gls{CBF} framework for general nonlinear systems with bounded uncertainty. We showed that the safety value function characterized by the dynamic programming Isaacs equation is a valid robust \gls{DCBF} protecting the \emph{maximal robust safe set}. Incorporating insights from \gls{RL}, we lifted this value function into state--control--disturbance space to derive a novel robust $\qfunc$-\gls{CBF} constraint for runtime safety filtering. Coupled with reachability-based adversarial \gls{RL}, the framework enables robust $\qfunc$-\gls{CBF} synthesis and deployment \emph{requiring only black-box access} to a transition mechanism, without closed-form dynamics, control-affine assumptions, or known uncertainty structure. We validated the framework on a disturbed inverted pendulum and a simulated 36-D quadruped. 
Our results provide a practical recipe for scalable synthesis and deployment of certifiable robust \gls{CBF} safety filters on high-dimensional nonlinear systems, using neural approximators that can be further strengthened through post-hoc verification.

\printbibliography

\end{document}